%% file: main.tex
\documentclass[11pt,a4paper]{article}
\usepackage{amsmath,amsfonts}
\usepackage{amsthm}
\usepackage{fullpage}
\usepackage{hyperref}
\usepackage{booktabs} 
\usepackage[ruled]{algorithm2e} 

\SetAlFnt{\small}
\SetAlCapFnt{\small}
\SetAlCapNameFnt{\small}
\SetAlCapHSkip{0pt}
\IncMargin{-\parindent}
\usepackage{graphicx,epsfig,amsmath,latexsym,amssymb,verbatim,algorithmic}

\usepackage{mathtools}
\usepackage{algorithmic}
\usepackage{array}
\usepackage{thm-restate}
\usepackage{bbm}
\usepackage{nicefrac}
\usepackage{algorithmic}
\usepackage[font={small,it}]{caption}
\usepackage{graphicx,epsfig,amsmath,latexsym,verbatim}
\usepackage{thmtools,thm-restate}
\usepackage{enumitem}
\usepackage{mdframed}
\usepackage{xcolor}

\newtheorem{theorem}{Theorem}

\newtheorem{corollary}[theorem]{Corollary}

\newtheorem{definition}[theorem]{Definition}

\newtheorem{lemma}[theorem]{Lemma}

\newcommand{\SW}{\mathrm{SW}}

\newcommand{\pro}{\textup{\textsc{OptimalSequence}}}

\newcommand{\Sn}{\mathcal{S}_{[n]}}
\newcommand{\Sc}{\mathcal{S}_C}
\newcommand{\Snoti}{\mathcal{S}_{-i}}

\newcommand{\Exp}{\mathbb{E}}
\newcommand{\R}{\mathbb{R}}
\newcommand{\poly}{\textup{poly}}

\newcommand{\maxF}{\textup{\textsc{MaxF}}}

\newcommand{\osf}{\textup{\textsc{OSF}}}

\newcommand{\osm}{\textup{\textsc{OSM}}}
\newcommand{\osa}{\textup{\textsc{OSA}}}
\newcommand{\oss}{\textup{\textsc{OSS}}}
\newcommand{\osi}{\textup{\textsc{OSI}}}

\newcommand{\satas}{\textup{\textsc{Sat-AS}}}
\newcommand{\bv}{{\bf v}}
\newcommand{\BR}{\textsc{BR}}

\DeclareMathOperator*{\argmin}{arg\,min}
\DeclareMathOperator*{\argmax}{arg\,max}

\title{\bf Optimizing over Serial Dictatorships}

\author{Ioannis Caragiannis \and Nidhi Rathi}
\date{Department of Computer Science, Aarhus University\\
{\AA}bogade 34, 8200 Aarhus N, Denmark\\
Email: \url{{iannis,nidhi}@cs.au.dk}}

\begin{document}


\maketitle


\input{abstract}

\input{intro}

\input{notation}

\input{general-bounds}
\input{specific-problems}
\input{truthful-implementation}

\bibliography{osd}
\end{document}

%% file: abstract.tex
\begin{abstract}
Motivated by the success of the {\em serial dictatorship} mechanism in social choice settings, we explore its usefulness in tackling various combinatorial optimization problems. We do so by considering an abstract model, in which a set of agents are asked to {\em act} in a particular ordering, called the {\em action sequence}. Each agent acts in a way that gives her the maximum possible value, given the actions of the agents who preceded her in the action sequence. Our goal is to compute action sequences that yield approximately optimal total value to the agents (a.k.a., {\em social welfare}). We assume {\em query access} to the value $v_i(S)$ that the agent $i$ gets when she acts after the agents in the ordered set $S$. 

We establish tight bounds on the social welfare that can be achieved using polynomially many queries. Even though these bounds show a marginally sublinear approximation of optimal social welfare in general, excellent approximations can be obtained when the valuations stem from an underlying combinatorial domain. Indicatively, when the valuations are defined using bipartite matchings, arborescences in directed graphs, and satisfiability of Boolean expressions, simple query-efficient algorithms yield $2$-approximations. We discuss issues related to truthfulness and show how some of our algorithms can be implemented truthfully using VCG-like payments. Finally, we introduce and study the {\em price of serial dictatorship}, a notion that provides an optimistic measure of the quality of combinatorial optimization solutions generated by action sequences.
\end{abstract}

%% file: intro.tex
\section{Introduction}
{\em Serial dictatorship} (SD) is arguably the most straightforward algorithm with numerous applications in many different settings~\cite{AS98,BM01,CFF+16,krysta2014size,RS90,FFZ:14,SVE:99}. In the economics literature (e.g., see~\cite{AS98}), it was first considered as a solution to the {\em house allocation} problem, where a set of houses have to be matched to agents who have strict preferences for them. The algorithm considers the agents in a fixed order and assigns the most-preferred available house to each agent considered.

Crucially, SD can be used in a much more general context. Consider a scenario with a set of agents, each having a menu of available actions and preferences for them. The objective is to compute a set of collective actions, with one action per agent. SD can perform the task by asking the agents to act in a predefined order. Whenever it is an agent's turn to act, she does so by selecting the action that is best possible, given the actions of agents who acted before. Previous actions may render some action of an agent unavailable and they can furthermore change her initial preferences.

This is a paradigm with several advantages. From the computational point of view, SD is as simple as possible. Furthermore, it is easily explainable to the agents, who are expected to have a minimal interaction with it. As SD never uses any input from the agents besides their actions, they have no other option but to act truthfully, selecting their best available action. On the negative side, the predefined order in which the agents are considered may lead to inefficiencies~\cite{CFF+16}. An agent who acts early may block actions which might be highly beneficial to agents acting later. {\em Random serial dictatorship}~\cite{AS98,BM01,FFZ:14,krysta2014size} is an attempt to address this inefficiency without deviating from the main principles of SD. The main idea is to consider the agents in a uniformly random order in the hope that problematic orderings will be used only rarely. However, inefficiencies may still happen~\cite{FFZ:14}.

We attempt a radical improvement of SD and instead ask: what is the ordering of the agents which SD must use in order to obtain the {\em best} possible results? In other words, we are interested in {\em optimizing} over serial dictatorships for a given problem. Of course, this task cannot be performed without using the agent preferences and, as such, creates apparent incentives to the agents to misreport their preferences in order to manipulate the outcome of the optimization to their benefit.

To mitigate this, we seek algorithmic solutions that use agent input sparingly. Starting with no information about the problem at hand, our algorithms learn (part of) the preferences by {\em posing queries} to the agents. A typical query to an agent asks for the value of her best possible action in an agent ordering. Our objective is to study the {\em query complexity}~\cite{DV12} of optimizing over serial dictatorships. Hence, we seek for algorithms that use as few queries as possible.

\subsection{Our contribution} At the highest level of abstraction, we define the problem \pro\ that captures the essence of optimizing over serial dictatorships. Instances of \pro\ consist of a set of $n$ agents. Each agent $i$ has a valuation function which returns the value $v_i(S)$ she gets by her best available action after the agents in the sequence $S$ have acted before $i$. Our aim is to design algorithms that compute an agent ordering that gives high value to the agents. Our benchmark is the {\em social welfare}, i.e., the total value of all agents. An important restriction is low query complexity; we seek for algorithms that use only polynomially many value queries. Our main result (under a mild monotonicity assumption) is an asymptotically tight bound of $\Theta\left(\frac{n\ln\ln{n}}{\ln{n}}\right)$ on the best approximation ratio that can be achieved by such algorithms that use randomization.

Next, we consider refined \pro\ instances that may be relevant for practical applications. We do so by assuming an underlying combinatorial structure which defines the valuations. Indicatively, we present three specific valuation structures, motivated by matching markets, network design, and constraint satisfaction. These define three special cases of \pro, which we call \osm, \osa, and \oss. In \osm, the agents have a set of items as their available actions. Whenever it is an agent's turn to act, she picks the available item that is of highest value to her. In \osa, the agents correspond to the nodes of a network and they can establish links to other nodes as actions. Whenever it is an agent's turn to act, she draws a link of maximum value under the constraint that the network stays acyclic.  In \oss, each agent controls a Boolean variable and has two actions: setting her variable to either True or False. These variables appear in a collection of weighted clauses. Whenever it is an agent's turn, she acts so that the weight of new satisfied clauses is maximized.

It turns out that \osm, \osa, and \oss\ admit much better algorithms. We present simple $2$-approximation algorithms for the three of them and leave open the question of better approximation results or lower bounds on the query complexity. We furthermore study computational questions related to the existence of serial dictatorships that produce a given solution to the underlying combinatorial problem. We prove that whether a given perfect matching for \osm\ and a given arborescence for \osa\ can be produced by serial dictatorships can be decided in polynomial time. In contrast, for \oss, we show that deciding whether a given Boolean assignment can be produced by a serial dictatorship is an NP-hard problem.

We also consider serial dictatorship as a {\em solution concept} and quantify the restrictions it poses on the solutions of optimization problems by introducing the notion of the {\em price of serial dictatorship} (PoSD). The PoSD of a combinatorial optimization problem is the worst-case ratio over all instances of the benefit of the optimal solution over the best social welfare obtained by a serial dictatorship. We prove optimal bounds of $1$ for both maximum-weight bipartite matching and maximum-weight arborescence in directed graphs. In contrast, for maximum satisfiability, the PoSD is between $3/2$ and $2$.

Finally, we scratch the surface of truthful implementations of our algorithms. We prove that our best algorithm for general instances of \pro\ has a truthful implementation through VCG-like payments (see~\cite{C71}) which require only polylogarithmically many queries to compute. Our algorithms for \osm\ and \osa\ are shown to be non-truthful. For \osa, we present a very simple randomized truthful mechanism that uses no payments.

Our work reveals many open problems; we give a list with the most important ones later. We remark that the approach can be naturally extended to many other combinatorial problems than the ones considered here.

\subsection{Further related work} Serial dictatorship and its variations have been well-studied in the matching literature~\cite{M13}. It seems however that the approach we consider here has not received attention therein. As an attempt to optimize serial dictatorships, one can view the notion of greedy weighted matchings in~\cite{DMS17}. A greedy weighted matching is one that can be produced by an algorithm which starts from an empty matching and iteratively puts an edge of maximum weight that is consistent with it. This is a restricted way to optimize over serial dictatorships. Furthermore, the authors in~\cite{DMS17} study the computational complexity of the problem assuming that the graph information (i.e., the edge weights) is given to the algorithm upfront. The notion of picking sequence in fair division \cite{gourves2021fairness} is close to the action sequences that we consider here; the crucial difference is that each agent can appear several times in a picking sequence.

Optimizing over serial dictatorships can be thought of as a particular way of exploiting greediness in computation. There have been several attempts to formally define greedy algorithms, including their relations to matroids, which are covered in classical textbooks on algorithm design~\cite{CLRS01}. In relation to combinatorial optimization, the work on incremental priority algorithms~\cite{BNR03,BBL+10,DI04} has conceptual similarities. Furthermore, for maximum satisfiability, there is a ongoing research line (e.g., see~\cite{PSW+17}) that aims to design simple greedy-like algorithms that achieve good approximation ratios. However, all these studies neglect query complexity questions. 

Finally, we remark that value queries, similar in spirit to those we consider here, have been used to solve the maximum welfare problem in combinatorial auctions~\cite{LLN06}. The PoSD notion is inspired by the price of stability in strategic games~\cite{ADK+08} and the price of fairness in allocation problems~\cite{BFT11,CKKK12}.

\subsection{Roadmap} The rest of the paper is structured as follows. We give preliminary definitions and our notation in Section~\ref{section:notations}. We then study the query complexity of \pro\ in Section~\ref{sec:general-instances}. The refinement of \pro\ with specific valuation structures is presented in Section~\ref{sec:val-struct}. The three case studies on bipartite matchings, arborescences, and satisfiability follow in Sections~\ref{sec:matching},~\ref{sec:arborescence}, and~\ref{sec:sat}. Issues related to truthful implementation are considered in Section~\ref{sec:truth}.

%% file: notation.tex
\section{Preliminaries}
\label{section:notations}
We consider scenarios where a set of agents are asked to \emph{act} in a particular order. Whenever it is the turn of an agent to act, she selects an action that gives her the maximum possible value, given the actions of agents who acted before her. In this section, we present basic definitions and notation for an abstract model that captures such scenarios. Later, in Section~\ref{sec:val-struct}, we refine this model further to bring it closer to well-known combinatorial optimization problems.

We denote by $n$ the number of agents, which we assume to be identified by the integers in set $[n]=\{1,2, ..., n\}$. An {\em action sequence} (respectively, {\em action subsequence}) $S$ is an ordering $S=(S(1),S(2), ...)$ of the agents in $[n]$ (respectively, some of the agents in $[n]$). For a subset of agents $C\subseteq [n]$, we denote by $\Sc$ the set of all action (sub)sequences of the agents in $C$. For an agent $i \in [n]$, we write $\Snoti$ to denote all possible action subsequences consisting of agents from the set $[n] \setminus \{i\}$. We reserve the letters $S$ and $\pi$ for use as action (sub)sequences.

For two action subsequences $S$ and $S'$, we write $S' \leq S$ if $S'$ is an ordered subset of $S$, i.e., if all agents in $S'$ are also in $S$ and have the same relative ordering. We use $S||i$ to denote the action (sub)sequence which is obtained by appending agent $i$ to the action subsequence $S$. We use $S \setminus j$ to denote the action subsequence obtained by removing agent $j$ from $S$, and preserving the order of the remaining agents. For an action (sub)sequence $S$, we write $S^i$ to denote the action subsequence consisting of the ordered set of agents that are ranked ahead of agent $i$ in $S$. We sometimes use the term {\em prefix} to refer to agents at the beginning of an action (sub)sequence. Clearly, for any agent $i$ who appears in action subsequence $S$, it is $S^i\leq S$. For an integer $c\in [n]$, we denote by $S|_c$ the action subsequence consisting of the first $c$ agents of $S$.

Every agent $i \in [n]$ has a {\em valuation function} $v_i: \Snoti \mapsto \R_{\geq 0}$ which, for every action subsequence $S\in \Snoti$, returns the value $v_i(S)$ that agent $i$ gets when acting immediately after the agents in the action subsequence $S$. We say that the valuations of an agent $i$ are {\em monotone} if for any two action subsequences $S,S' \in \Snoti$, $S'\leq S$ implies $v_i(S')\geq v_i(S)$. The {\em social welfare} $\SW(S)$ of an action sequence $S$ is the total value of the agents in the sequence, i.e., $\SW(S) = \sum_{i\in [n]}{v_i(S^i)}$. 

We study the problem of computing an action sequence having maximum social welfare, which we call \pro. In particular, an instance of \pro\ consists of the valuation functions of $n$ agents and the objective is to compute an action sequence of maximum social welfare. An important limitation of our model is that the valuation functions are not given to the algorithm as part of the input. Instead, the algorithms for solving \pro\ can access the entry $v_i(S)$ of the valuation function through a query. We are interested in algorithms that solve \pro\ using a polynomial number of queries. In other words, we are interested in understanding the {\em query complexity} of \pro. Notice that exponentially many queries (for each agent $i$ and each action subsequence in $\Snoti$) can always be used to learn the valuation functions. Then, \pro\ can be solved exactly. We remark that computational limitations are not our concern in this work and we assume to have unlimited computational resources to process the valuations once these are available.

As we will see, solving \pro\ exactly using polynomially many queries is a challenge. So, we are interested in algorithms that compute approximately optimal solutions. On input an \pro\ instance, a $\rho$-{\em approximate} action sequence (for $\rho\geq 1$) is one which has a social welfare that is at most $\rho$ times smaller than the maximum possible (or optimal) social welfare over all action sequences. We refer to $\rho$ as the {\em approximation ratio} of an algorithm which always returns $\rho$-approximate action sequences. We are interested in \pro\ algorithms that make only polynomially many queries and achieve an approximation ratio that is as low as possible.

%% file: general-bounds.tex
\section{The Query Complexity of Approximate Serial Dictatorships}\label{sec:general-instances}
We devote this section to studying the query complexity of approximation algorithms for \pro, presenting both positive and negative results. It is not hard to see that non-monotone valuations are notoriously hard to cope with and any finite approximation may require to examine almost all the $n!$ action sequences.\footnote{Indeed, consider instances in which a special agent $i$ has valuation $v_i(S)=1$ if $S$ contains all $n-1$ other agents in a specific hidden order and all other agent valuations are zero. To compute the only action sequence with non-zero social welfare, an algorithm needs to ``guess'' the hidden order.} However, as we will see later in Section~\ref{sec:sat}, excellent approximations for special cases of non-monotone valuations are possible.

So, in the rest of this section, we focus on instances with monotone valuations. On the positive side, we present two algorithms for \pro, one deterministic and one randomized, called {\sc Det} and {\sc Rand}, respectively. Both algorithms use an integer parameter $c\geq 1$ and apply to $n$-agent \pro\ instances. We denote by $N_c$ the collection consisting of all sets of $c$ distinct agents from $[n]$. Algorithm {\sc Det} considers all $c$-sets in $N_c$. For each set $C\in N_c$, it considers every action subsequence $S\in \Sc$ and selects the one, call it $S'$, that maximizes the quantity $\sum_{i\in C}{v_i(S^i)}$. Finally, it returns an action sequence that contains $S'$ as a prefix (by filling the last $n-c$ positions in the action sequence with agents in $[n]\setminus C$ in arbitrary order).

\begin{theorem}\label{thm:det}
On input $n$-agent monotone \pro\ instances, algorithm {\sc Det} makes ${n \choose c}\cdot c\cdot c!$ queries and has an approximation ratio of at most $\frac{n}{c}$.
\end{theorem}

\begin{proof}
Consider an $n$-agent \pro\ instance and notice that the algorithm enumerates all the ${n \choose c}$ elements of $N_c$ and for each such element $C\in N_c$ and each of the $c!$ action subsequences  $S\in \Sc$, it makes $c$ queries to compute the value of the sum $\sum_{i\in C}{v_i(S^i)}$. The bound on the number of queries follows.

Given an $n$-agent \pro\ instance, let $\widetilde{\pi}$ be the optimal action sequence, $\pi$ be the action sequence returned by algorithm {\sc Det}, and $S=\pi|_c$. Denote by $\widetilde{S}\leq \widetilde{\pi}$ the action subsequence consisting of those $c$ agents that have the $c$ highest values $v_i(\widetilde{\pi}^i)$ (breaking ties arbitrarily). Then, by the definition of {\sc Det}, 
\begin{align}\label{eq:thm:det-1}
\sum_{i\in S}{v_i(S^i)} &\geq \sum_{i\in \widetilde{S}}{v_i(\widetilde{S}^i)}.
\end{align}
Furthermore, by monotonicity, the fact that $\widetilde{S}^i\leq \widetilde{\pi}^i$ implies that $v_i(\widetilde{S}^i)\geq v_i(\widetilde{\pi}^i)$ for every agent $i$ and, thus,
\begin{align}\label{eq:thm:det-2}
\sum_{i\in \widetilde{S}}{v_i(\widetilde{S}^i)} &\geq \sum_{i\in \widetilde{S}}{v_i(\widetilde{\pi}^i)}.
\end{align}
Finally, by the definition of $\widetilde{S}$, we have
\begin{align}\label{eq:thm:det-3}
    \sum_{i\in 
    \widetilde{S}}{v_i(\widetilde{\pi}^i}) &\geq \frac{c}{n}\cdot \sum_{i\in [n]}{v_i(\widetilde{\pi}^i)}.
\end{align}
By inequalities (\ref{eq:thm:det-1}), (\ref{eq:thm:det-2}), (\ref{eq:thm:det-3}), and the definition of the social welfare, we get
\begin{align*}
    \SW(\pi) &= \sum_{i\in [n]}{v_i(\pi^i)} \geq \sum_{i\in S}{v_i(S^i)}\geq \frac{c}{n} \cdot \sum_{i\in [n]}{v_i(\widetilde{\pi}^i)} = \frac{c}{n} \cdot \SW(\widetilde{\pi}).
\end{align*}
The bound on the approximation ratio follows.
\end{proof}
Theorem~\ref{thm:det} implies the following using $c=1/\epsilon$ for constant $\epsilon>0$.

\begin{corollary}
For every constant $\epsilon>0$, on input $n$-agent monotone \pro\ instances, algorithm {\sc Det} has an approximation ratio at most $\epsilon\cdot n$ by making $O\left(n^{2/\epsilon}\right)$ queries.
\end{corollary}

Algorithm {\sc Rand} is similar to {\sc Det}; it uses randomization to avoid the enumeration over all elements of $N_c$. First, it selects a $c$-set $C$ from $N_c$ uniformly at random. Then, it considers every action subsequence $S\in \Sc$ and selects the one, call it $S'$, that maximizes the quantity $\sum_{i\in C}{v_i(S^i)}$. Finally, it returns an action sequence that contains $S'$ as a prefix.

\begin{theorem}\label{thm:rand}
On input $n$-agent monotone \pro\ instances, algorithm {\sc Rand} makes $c\cdot c!$ queries and has an approximation ratio of at most $\frac{n}{c}$.
\end{theorem}

\begin{proof}
The bound on the number of queries is due to the fact that {\sc Rand} considers all the $c!$ action subsequences with the agents in $C$ and computes the sum $\sum_{i\in C}{v_i(S^i)}$ by making $c$ queries for each of them.

To bound the approximation ratio, consider an $n$-agent \pro\ instance and let $\widetilde{\pi}$ and $\pi$ be the optimal action sequence and the (random) action sequence returned by {\sc Rand}, respectively. Also, let $S_C$ be the action subsequence consisting of agents in $C$ so that $S_C\leq \widetilde{\pi}$. Since {\sc Rand} prefers the prefix $\pi|_c$ to $S_C$, we have
\begin{align}\label{eq:thm:rand-1}
    \sum_{i\in C}{v_i(\pi^i)} &\geq \sum_{i\in C}{v_i(S_C^i)} \geq \sum_{i\in C}{v_i(\widetilde{\pi}^i)}.
\end{align}
The last inequality follows by monotonicity of valuations, which implies that $v_i(S_C^i)\geq v_i(\widetilde{\pi}^i)$ since $S_C^i\leq \widetilde{\pi}^i$ for every agent $i\in C$. Now, since $C$ is selected uniformly at random from $N_c$, an agent from $[n]$ belongs to $C$ with probability $\frac{c}{n}$. Using linearity of expectation, inequality (\ref{eq:thm:rand-1}), and this last observation, we have
\begin{align*}
    \Exp\left[\SW(\pi)\right]&\geq \Exp\left[\sum_{i\in C}{v_i(\pi^i)}\right] \geq \Exp\left[\sum_{i\in C}{v_i(\widetilde{\pi}^i)}\right] = \frac{c}{n}\cdot \Exp\left[\sum_{i\in [n]}{v_i(\widetilde{\pi}^i)}\right]=\frac{c}{n}\cdot \Exp\left[\SW(\widetilde{\pi})\right].
\end{align*}
This completes the proof.
\end{proof}

For $c=\Theta\left(\frac{\ln{n}}{\ln\ln{n}}\right)$, Theorem~\ref{thm:rand} yields the following.

\begin{corollary}
On input $n$-agent monotone \pro\ instances, algorithm {\sc Rand} makes at most $\poly(n)$ queries and has an approximation ratio of at most $O\left(\frac{n\ln\ln{n}}{\ln{n}}\right)$.
\end{corollary}

In the following, we show that {\sc Rand} is (asymptotically) optimal.
\begin{theorem} \label{theorem:lower-bound}
For every integer $c\geq 1$, any (possibly randomized) algorithm for monotone \pro\ that makes at most $\frac{c!}{n}-1$ queries when applied on instances with $n$ agents has an approximation ratio at least $\frac{n}{c+1}$.
\end{theorem}

\begin{proof}
For an integer $c\geq 1$, we will define a probability distribution $Q_c$ over $n$-agent instances of \pro\ with an optimal social welfare of $n$, and will prove that the expected social welfare of any deterministic algorithm that makes at most $\frac{c!}{n}-1$ queries on these instances is at most $c+1$. Then, by Yao's principle~\cite{Y77}, the quantity $\frac{n}{c+1}$ will be a lower bound on the approximation ratio of randomized algorithms for \pro.

We define the family $\mathcal{F}_c$ consisting of instances $I_{\pi}$ for every $\pi\in \Sn$. Instance $I_{\pi}$ is defined by the following valuation functions. For any agent $i\in [n]$ and any action subsequence $S$ from $\Snoti$:
	\[
	v_i(S) =
	\left\{
	\begin{array}{ll}
		1  & \mbox{if }   |S|<c  \ \ \text{or} \ \ S \leq \pi \\
		0 & \mbox{otherwise } 
	\end{array}
	\right.
	\]
Hence, the valuation functions are binary. An agent gets a value of $1$ only in the following two cases: (a) when she has fewer than $c$ agents ahead of her, or (b) when the agents ranked ahead of her have the same relative ordering as in the action sequence $\pi$. In all other cases, her value is $0$. Notice that the action sequence $\pi$ has optimal social welfare of $n$ for instance $I_{\pi}$.

Let us prove that the instances in $\mathcal{F}_c$ are indeed valid ones, i.e., the valuations defined are monotone. Fix any agent $i \in [n]$ and any action subsequence $S \in \Snoti$. To prove monotonicity, we need to show that $v_i(S') \geq v_i(S)$ for any action subsequence $S' \leq S$. This follows trivially when $|S|<c$ or when $v_i(S)=0$. Therefore, let us assume that $|S| \geq c$ and $v_i(S) =1$. This implies that $S$ is an action subsequence of $\pi$, i.e., $S\leq \pi$. By transitivity, we have $S' \leq \pi$ as well, and the definition of $v_i$ yields $v_i(S')=1$, thereby proving that $v_i$ is monotone.

For integer $c\geq 1$, the probability distribution $Q_c$ simply selects a {\em magic} action sequence $\widehat{\pi}$ and produces the instance $I_{\widehat{\pi}}$. Consider a deterministic algorithm $\mathcal{A}$ that makes at most $k=\frac{c!}{n}-1$ queries; we will show that the expected social welfare of its output on the (random) instances produced by distribution $Q_c$ is at most $c+1$. Assume that, on input the instance $I_{\widehat{\pi}}$, algorithm $\mathcal{A}$ makes $k$ queries of the form $(i,S)$ and eventually returns the action sequence $\overline{\pi}\in \Sn$. We say that an action (sub)sequence $S$ {\em hits} the magic action sequence $\widehat{\pi}$ when $|S|\geq c$ and $S|_c\leq \widehat{\pi}$. We define the {\em enhanced} version of $\mathcal{A}$, denoted by $\mathcal{A}_{\widehat{\pi}}$, that works as follows when applied on instance $I_{\widehat{\pi}}$. Algorithm $\mathcal{A}_{\widehat{\pi}}$ simulates algorithm $\mathcal{A}$, performing the following additional steps. Whenever $\mathcal{A}$ makes a query $(i,S)$ so that $S$ hits the magic action sequence $\widehat{\pi}$, $\mathcal{A}_{\widehat{\pi}}$ returns $\widehat{\pi}$ and terminates. Also, if the action sequence $\overline{\pi}$ returned by $\mathcal{A}$ hits $\widehat{\pi}$, algorithm $\mathcal{A}_{\widehat{\pi}}$ returns $\widehat{\pi}$ instead. Otherwise, it returns $\overline{\pi}$.

Observe that, the social welfare of algorithm $\mathcal{A}$ when applied on the instance $I_{\widehat{\pi}}$ is never higher than the social welfare of the enhanced algorithm $\mathcal{A}_{\widehat{\pi}}$. This is due to the fact that $\mathcal{A}_{\widehat{\pi}}$ returns a suboptimal action sequence (i.e., one that is different than $\widehat{\pi}$) only when this sequence is returned by $\mathcal{A}$. Hence, it suffices to bound the social welfare of algorithm $\mathcal{A}_{\widehat{\pi}}$ when applied on instance $I_{\widehat{\pi}}$.

For an action (sub)sequence $S$, observe that $\Sn$ contains $\frac{n!}{c!}$ action sequences $\pi'$ such that $S|_c\leq \pi'$. Since $\widehat{\pi}$ is selected uniformly at random among the $n!$ action sequences of $\Sn$, the probability that $S$ hits $\widehat{\pi}$ is $\frac{1}{c!}$. By applying the union bound, we obtain that the probability that either some of the $k$ action subsequences that algorithm $\mathcal{A}$ uses in its $k$ queries or the action sequence $\overline{\pi}$ it outputs hit the magic action sequence $\widehat{\pi}$ is at most $\frac{k+1}{c!}$. We conclude that the enhanced algorithm $\mathcal{A}_{\widehat{\pi}}$ returns $\widehat{\pi}$ with probability at most $\frac{k+1}{c!}$ and gets social welfare $n$. Notice that if this is not the case, the social welfare of the action sequence returned is exactly $c$. Then, the expected welfare of algorithm $\mathcal{A}$ is
\begin{align*}
    \Exp_{I\sim Q_c}[\SW(\mathcal{A(I)})] &\leq \Exp_{\widehat{\pi}\sim Q_c}[\SW(\mathcal{A}_{\widehat{\pi}}(I_{\widehat{\pi}}))]\leq c+(n-c)\cdot\frac{k+1}{c!}\leq c+1,
\end{align*}
as desired. The last inequality follows by the definition of $k$.
\end{proof}

By applying Theorem~\ref{theorem:lower-bound} with $c=\Theta\left(\frac{\ln{n}}{\ln\ln{n}}\right)$, we obtain that {\sc Rand} is essentially best possible.
\begin{corollary}
Any (possibly randomized) algorithm for monotone \pro\ that makes $\poly(n)$ queries has approximation ratio $\Omega\left(\frac{n\ln\ln{n}}{\ln{n}}\right)$.
\end{corollary}

%% file: specific-problems.tex
\section{Specific Valuation Structures} \label{sec:val-struct}
We now explore specific instances of \pro, in which valuations have a combinatorial structure. In particular, the definition of valuations assumes that each agent $i$ has a set $X_i$ from which she can select her {\em actions}. A {\em collection of actions} is just a set of agent-action pairs $(i,a)$, each consisting of an agent $i$ and an action $a\in X_i$. An agent $i$ may appear in at most one agent-action pair of a collection. A collection of actions that includes all agents is called {\em full}. A constraint $F$ defines which collections of actions are {\em feasible}. $F$ is called {\em downward closed} if every subset of a feasible collection of actions is feasible as well.

An action (sub)sequence $S$ produces a collection of actions $A$ as follows. Initially, $A$ is empty. When it is the turn of agent $i$ to act, she selects an action $\widehat{a}$ from $X_i$ which is best possible among all actions $a$ in $X_i$ such that $A\cup\{(i,a)\}$ is feasible. The agent-action pair $(i,\widehat{a})$ is then included in the collection of actions $A$. To decide which action is best possible, agent $i$ uses a strict ranking of the actions in $X_i$ (possibly dependent on other agent-action pairs in $A$ when this decision is made) and corresponding valuations for them that are consistent to it. We call the ranking of an agent and the corresponding valuations for their action {\em endogenous} when it is independent of actions by other agents. For an agent $i$ who acts after a collection of actions by other agents, we denote her best possible action as $\BR(i,A)$. 

We consider three different feasibility constraints, which define three different classes of valuations and corresponding restricted classes of \pro\ instances; others can be defined analogously. Essentially, with each class, we restrict our attention to \pro\ instances that may appear in practical applications and which are hopefully easier to solve than the very general ones. The three valuation structures we consider are described in the next three subsections.

\subsection{Matchings in bipartite graphs} The valuations are defined using a complete bipartite graph $G=K_{n,n}$, equipped with a function $w:[n]^2\rightarrow \R_{\geq 0}$, that assigns non-negative weights to its edges. The nodes in the left part of the bipartition correspond to the $n$ agents. The nodes in the right part correspond to items, which are to be given to the agents, under the restriction that each agent can get (at most) one item and each item will be given to (at most) one agent. Such an assignment of items to agents is represented by a {\em matching} in $G$. The weight $w(i,j)$ denotes the value agent $i$ has for item $j$. Agents pick their items according to an action sequence. When it is the turn of an agent to act, she picks the item of highest value among those that have not been picked by agents who acted before her.
    
Formally, each agent $i$ has a strict ranking $r_i$ that ranks the items in monotone non-increasing order with respect to the value of agent $i$ for them, breaking ties in a predefined (agent-specific) manner. We denote by $r_i(j)$ the rank of item $j$ for agent $i$. The items of rank $1$ and $n$ for agent $i$ are ones for which she has highest and lowest values, respectively. Also, $r_i(j)<r_i(j')$ implies that $w(i,j)\geq w(i,j')$. For an action subsequence $S\in \Snoti$, denote by $\mu_i(S)$ the item agent $i$ gets when she acts immediately after the agents in $S$. By overloading notation, we use $\mu(S)$ to denote the set of items picked by the agents who acted before agent $i$, i.e., $\mu(S)=\cup_{k\in S}{\mu_k(S^k)}$. Clearly, $\mu(\emptyset)=\emptyset$. Then, the action of agent $i$ when acting after the action subsequence $S\in \Snoti$ is to pick item $\mu_i(S)= \argmin_{j\in [n]\setminus \mu(S)}{r_i(j)}$, which gives her value $v_i(S)=w(i,\mu_i(S))$. Notice that, when all agents have acted, the items picked form a {\em perfect matching} in $G$.

We denote by \osm\ the special cases of \pro\ when applied on instances with the specific valuation structure as defined above.
To make a connection with the general combinatorial structure described above, each agent $i$ has the set of all items as its set of available actions $X_i$. The feasibility constraint $F$ requires that the collection of actions form a matching in $G$. Clearly, $F$ is downward closed. All agents have endogenous rankings/valuations.

\begin{lemma}
\osm\ valuations are monotone.
\end{lemma}
    
\begin{proof}
We will prove the claim by showing the following property: for every agent $i\in [n]$ and action subsequences $S,\widehat{S}\in \Snoti$ with $\widehat{S}\leq S$, it holds $\mu(\widehat{S})\subseteq \mu(S)$. This implies that $r_i(\mu_i(\widehat{S}))\leq r_i(\mu_i(S))$ and, consequently, $v_i(\widehat{S})\geq v_i(S)$.
        
To complete the proof, assume that the property is not true. Then, there exists an agent $k$ who gets in $\widehat{S}$ an item that does not belong to set $\mu(S^k)$ and is different than item $\mu_k(S^k)$. Assume, furthermore, that among the agents satisfying these conditions, $k$ is the agent that acts earliest in $\widehat{S}$. This would mean that the two items $\mu_k(\widehat{S}^k)$ and $\mu_k(S^k)$ are both available to agent $k$ when his turn in $\widehat{S}$ and $S$ comes. Picking a different item in $\widehat{S}$ and $S$ violates our definition of agent actions.
\end{proof}
    
\subsection{Arborescences in directed graphs} The valuations are defined using an $n$-node complete directed graph $G=([n],E)$ with a weight function $w:E\rightarrow \R_{\geq 0}$ on its edges. The nodes correspond to the agents. The weight $w(i,j)$ denotes the value of agent $i$ for the edge $(i,j)$ (i.e., the directed edge from node $i$ to node $j$). Each agent has as objective to connect to another node through a directed edge of maximum value so that this edge forms a {\em directed forest} with edges drawn before.

Formally, each agent $i$ has a strict ranking $r_i$ that ranks the outgoing edges from node $i$ in monotone non-increasing order with respect to the value of agent $i$ for them, breaking ties in a predefined manner. We denote by $r_i(j)$ the rank of edge $(i,j)$ for agent $i$. The edges of rank $1$ and $n-1$ for agent $i$ are ones for which agent $i$ has highest and lowest values, respectively. Also, $r_i(j)<r_i(j')$ implies that $w(i,j)\geq w(i,j')$. For an action subsequence $S\in \Snoti$, denote by $G_S$ the subgraph of $G$ consisting of the edges drawn during the actions of the agents in $S$. Clearly, $G_\emptyset$ consists of the $n$ nodes of $G$ and no edges. For an agent $i$ and action subsequence $S\in \Snoti$, we say an edge $(i,j)$ is {\em forbidden} if it forms a cycle when put together with the graph $G_S$. The action of agent $i$ when acting after the action sequence $S\in \Snoti$ is to draw an edge to node $\tau_i(S)$ defined as the node $j$ minimizing $r_i(j)$ under the constraint that edge $(i,j)$ is not forbidden. The agent draws no edge if all edges are forbidden and gets a value of $v_i(S)=0$. Notice that this will happen only for an agent who acts after all other agents. Otherwise, the value of agent $i$ is $v_i(S)=w(i,\tau_i(S))$. Notice that, when all agents have acted, the edges drawn form an {\em arborescence} in $G$, i.e., a tree whose edges are directed towards the root, which has out-degree $0$.

According to the terminology of the general combinatorial structure,
each agent $i$ has the set of all edges outgoing from node $i$ as its set of available actions $X_i$. The feasibility constraint $F$ requires that the collection of actions form a directed forest in $G$; this is clearly downward closed. Again, all agents have endogenous rankings/valuations. We use the term \osa\ to refer to these instances of \pro.

\begin{lemma}
\osa\ valuations are monotone.
\end{lemma}

\begin{proof}
Consider agent $i\in [n]$ and two action subsequences $S,\widehat{S}\in \Snoti$ such that $\widehat{S}\leq S$. We will show the following property: the edge agent $i$ draws when she acts after the action subsequence $S$ is not forbidden when she acts after the action subsequence $\widehat{S}$. Hence, we get $v_i(\widehat{S})\leq v_i(S)$, as desired. 

For the sake of contradiction, assume that the property is not true. Then, there should exist action subsequences $S$ and $\widehat{S}$ with $\widehat{S}\leq S$ and an agent $k$ in $\widehat{S}$, for which the edge $(k,j)$ she picks in $S$ is forbidden when her turn to act comes in $\widehat{S}$. Among all agents who may satisfy this condition, we specifically select $k$ to be the one who acts earliest in $\widehat{S}$. As edge $(k,j)$ is forbidden for agent $k$ in $\widehat{S}$, let $C$ be the cycle it forms together with $G_{\widehat{S}}$. Let $j_1$ be the node/agent that is furthest from $k$ in $C$ among those drawing different edges in $S$ and $\widehat{S}$. Let $(j_1,j_2)$ and $(j_1,j_3)$ be the edges agent $j_1$ draws in $\widehat{S}$ and $S$, respectively.

We now claim that when it is the turn of agent $j_1$ to act in $S$, the edge $(j_1,j_2)$ is not forbidden. The reason is that the only nodes to which node $j_1$ can be connected via edges that were drawn in $S$ before agent $j_1$ acts, are those in the path from $j_2$ to $k$ in $C$. As agent $k$ has not acted before agent $j_1$, no cycle can be formed if agent $j_1$ draws edge $(j_1,j_2)$. Hence, $r_{j_1}(j_2)>r_{j_1}(j_3)$. But then, agent $j_1$ who acts before agent $k$ in $\widehat{S}$ and draws edge $(j_1,j_2)$ must have the edge $(j_1,j_3)$ she drew in $S$ as forbidden. This contradicts our assumption on the timing of agent $k$.
\end{proof}

\subsection{Satisfiability of clauses} The valuations are defined using $n$ Boolean variables $x_1$, $x_2$, ..., $x_n$, a set $\mathcal{C}$ of $m$ clauses $C_1$, $C_2$, ..., $C_m$, and a function $w:\mathcal{C}\rightarrow \R_{\geq 0}$ that assigns a non-negative weight  to each clause. Agent $i\in[n]$ controls the variable $x_i$. Informally, when acting according to an action sequence, an agent sets the value of her variable to either True or False, so that the additional weight in the newly satisfied clauses is maximized. The value of an agent is then simply this additional weight of clauses that became satisfied due to her action.

Formally, for agent $i\in [n]$ and a set of clauses $M$, we denote by $P(i,M)$ (respectively, $N(i,M)$) the total weight of the clauses in $M$ which contain the positive literal $x_i$ (respectively, the negative literal $\overline{x_i}$). For an action subsequence $S$, we denote by $C_S$ the set of clauses that are not satisfied by the actions of the agents in $S$. Clearly, $C_{\emptyset}=\mathcal{C}$. The action of agent $i$ when acting after the agents in the action subsequence $S\in \Snoti$ is to set her variable $x_i$ to True if $P(i,C_S)>N(i,C_S)$ and to False if $P(i,C_S)<N(i,C_S)$. Ties, i.e., when $P(i,C_S)=N(i,C_S)$, are resolved by setting $x_i$ to a predefined value (True or False). Then, the valuation $v_i(S)$ is simply $P(i,C_S)$ or $N(i,C_S)$, whichever is higher, i.e., $v_i(S)=\max\{P(i,C_S),N(i,C_S)\}$.

 According to the terminology above, each agent $i$ has True and False as her two available actions in $X_i$ and there is no feasibility constraint. In contrast to the two classes above, the agent rankings/values of the two actions are non-endogenous. We denote by \oss\ the special cases of \pro\ when applied on satisfiability instances.

\begin{lemma}
\oss\ valuations are not monotone.
\end{lemma}

\begin{proof}
Consider the instance with three agents and the following clauses: $x_1\vee x_2$ of weight $6$, $\overline{x_1} \vee x_2 \vee x_3$ of weight $2$, $\overline{x_1} \vee \overline{x_2}\vee x_3$ of weight $1$, and $\overline{x_1}\vee \overline{x_2}$ of weight $2$. Consider the action subsequence $12$. Agent $1$ sets $x_1=1$ and then agent $2$ sets $x_2=0$, leaving $\overline{x_1}\vee x_2\vee x_3$ as the only unsatisfied clause. Hence, agent $3$ sets $x_3=1$ and gets a value of $v_3(12)=2$. Now, consider the action subsequence that consists only of agent $2$. Agent $2$ set $x_2=1$ and, then, the only clause that is unsatisfied and contains variable $x_3$ is $\overline{x_1} \vee \overline{x_2}\vee x_3$. We get $v_3(2)=1$, which violates monotonicity of $v_3$.
\end{proof}

\subsection{Three key questions} 
We devote the last part of this section for introducing the three {\em key questions} that we consider, using the term \osf\ as a proxy for a valuation structure defined by a feasibility constraint $F$. 

With specific valuation structures, we hope to be able to achieve considerably better approximations with polynomially many queries, compared to the upper bounds for general \pro\ instances in Theorem~\ref{thm:rand}. We remark that, like in our algorithms in Section~\ref{sec:general-instances}, the only tool that is available to algorithms for these problems is query access to the valuations. However, knowledge of the fact that the valuations have a specific underlying structure may prove useful.
\smallskip

\begin{quote}\begin{mdframed}
    {\bf Question 1:} How well can algorithms that use polynomially many queries approximate \osf?
\end{mdframed}\end{quote}
\smallskip

We remark here that the computational complexity of the underlying optimization problem does not have any implications for the query complexity of (approximating) the corresponding \pro\ instance. We demonstrate this with a tailored class of instances (defining the special case \osi\ of \pro) that use independent sets to define the valuations. Does the fact that the maximum independent set problem is NP-hard (even to approximate) imply any lower bound on the query complexity of \osi? We will see that the answer is negative. 

An \osi\ instance is defined with an $n$-node undirected graph $G$. Each node corresponds to an agent. For an agent $i\in [n]$, the value $v_i(S)$ for an action subsequence $S\in \Snoti$ is $1$ if the nodes in $S\cup \{i\}$ form an independent set and $0$ otherwise. It is not difficult to see that these valuations are monotone. Furthermore, the social welfare $\SW(\pi)$ of an action sequence $\pi$ is equal to the length of the maximal prefix of $\pi$ consisting of agents whose corresponding nodes form an independent set in $G$. Hence, the optimal social welfare among all action sequences is equal to the size of the maximum independent set. Now, notice that a simple algorithm that uses only $O(n^2)$ queries can learn the underlying graph $G$, compute a maximum independent set in it, and, consequently, an action sequence with maximum social welfare. It just suffices to query the value $v_i(j)$ for every pair of agents $i$ and $j$. Then, the edge $(i,j)$ exists in $G$ if $v_i(j)=0$ and does not exist if $v_i(j)=1$.

Next, we would like to understand how rich the space of action sequences is with respect to all feasible solutions to the underlying combinatorial problem.
\smallskip

\begin{quote}\begin{mdframed}
    {\bf Question 2:} What is the computational complexity of the following problem? Given a collection of actions that satisfies the feasibility constraint $F$, decide whether it can be produced by an action sequence.
\end{mdframed}\end{quote}
\smallskip

An action sequence can be viewed as a greedy algorithm for computing a solution to the combinatorial optimization problem \maxF, whose objective is to compute a feasible collection of actions of maximum total value for the agents. An important question about the quality of action sequences is then:
\smallskip

\begin{quote}\begin{mdframed}
    {\bf Question 3:} How well can the best action sequence approximate the optimal solution to the underlying $\maxF$ combinatorial optimization problem?
\end{mdframed}\end{quote}
\smallskip

\noindent We introduce the concept of the {\em price of serial dictatorship} (PoSD) of \maxF\ to quantify the answer to Question 3. The price of serial dictatorship of an instance of \maxF\ is the ratio of the maximum total value for the agents among all the feasible collections of actions over the social welfare of the best action sequence. Then, the price of serial dictatorship of \maxF\ is the worst-case PoSD over all \maxF\ instances. Note that PoSD is a number higher or equal to $1$.

\section{Matchings in Bipartite Graphs} \label{sec:matching}
We now present answers to these three key questions for \osm. An answer to Question 1 can be obtained using Algorithm~\ref{alg:matching-2-approx}, which is presented in the following.

 {
 	\begin{algorithm}[ht]
 		{ 

 			{\bf Input:} Query access to an $n$-agent \osm\ instance 

 			{\bf Output:} An action sequence $\pi$
 			\caption{A $2$-approximation algorithm for \osm}	
 			\label{alg:matching-2-approx}
 			\begin{algorithmic}[1]
 				\STATE Initialize $\pi\leftarrow \emptyset$; $R\leftarrow [n]$; 
 				\FOR{$k \leftarrow 1$ to $n$}
 				\STATE Let $i^*\leftarrow \argmax_{i\in R}{v_i(\pi)}$;
 				\STATE Update $\pi(k)\leftarrow i^*$; $R\leftarrow R\setminus \{i^*\}$;
 				\ENDFOR
 				\RETURN $\pi$;
   		\end{algorithmic}
 		}
 	\end{algorithm}
 }		

Algorithm~\ref{alg:matching-2-approx} builds the action sequence $\pi$ gradually in steps. In the $k$-th step, it selects the agent $i^*$ (among those in the set variable $R$, which contains the agents not included in $\pi$ yet) who attains maximum value if she acts immediately after the action subsequence built so far. Notice that $\pi$ is an empty action subsequence initially and is augmented in each execution of the for loop of Algorithm~\ref{alg:matching-2-approx}.

\begin{theorem}\label{thm:osm-2-approx}
    On input $n$-agent \osm\ instances, Algorithm~\ref{alg:matching-2-approx} uses $O(n^2)$ queries and computes a $2$-approximate solution.
\end{theorem}
\begin{proof}
The bound on the number of queries is obvious: the algorithm makes $n$ executions of the for loop and, in each of them, it makes a query $(i,\pi)$ for each of the at most $n$ agents in set $R$.  

Let $\widehat{\pi}$ be an optimal action sequence for a given \osm\ instance consisting of a complete bipartite graph $G = K_{n,n}$, edge-weight function $w:[n]^2\rightarrow \R_{\geq 0}$, and rank information $\{r_i\}_{i\in [n]}$ and $\pi$ the action sequence returned by Algorithm~\ref{alg:matching-2-approx}. Let $M$ and $\widehat{M}$ be the sets of edges defined by the items picked in action sequences $\pi$ and $\widehat{\pi}$. Let $e_i$ be the edge that is incident to node $i$ in $\widehat{M}$. Let $C$ be the set of agents who pick a different item in $\pi$ and $\widehat{\pi}$. For each agent $i\in C$, denote by $e_i^+$ and $e_i^-$ the two edges of $M$ which are incident to edge $e_i$. Notice that the definition of Algorithm~\ref{alg:matching-2-approx} implies that, for any $i\in C$, the agent who picks the first among edges $e_i^+$ and $e_i^-$ gets a value that is at least as high as the weight of edge $e_i$ (since $e_i$ was available at that point and agent $i$ had not picked any item yet). Hence, $w(e_i)\leq w(e_i^+)+w(e_i^-)$ and 
\begin{align*}
    \SW(\widehat{\pi}) &= \sum_{i\in [n]}{w(e_i)} = \sum_{i\in C}{w(e_i)}+\sum_{i\in [n]\setminus C}{w(e_i)} \leq \sum_{i\in C}{(w(e_i^+)+w(e_i^-))}+\sum_{i\in [n]\setminus C}{w(e_i)}\\
    &= 2\cdot \sum_{i\in C}{w(i,\mu_i(\pi^i))}+\sum_{[n]\setminus C}{w(i,\mu_i(\pi^i))}\leq 2\cdot \SW(\pi),
\end{align*}
as desired.
\end{proof}

Despite the simplicity of Algorithm~\ref{alg:matching-2-approx}, improvements seem to be technically challenging:

\smallskip

\begin{quote}\begin{mdframed}
{\bf Open Question 1:} Are there algorithms that solve \osm\ optimally using polynomially many queries?
\end{mdframed}\end{quote}

\smallskip

We suspect that the answer to Open Question 1 is positive, at least for the special case of \osm\ where each agent has different valuations for the items. Unfortunately, our attempts to design an algorithm that learns the underlying valuations in polynomial time have not been successful so far.
In our answer to Question 2, we will first use our more general terminology as we intend to provide an answer for a large class of valuation structures; our answer to Question 2 for \osm\ will follow as a corollary.

\begin{theorem}\label{thm:as-for-F}
Consider an \osf\ instance defined by a downward closed feasibility constraint $F$ and endogenous valuations, and let $A=\{(i,a_i)\}_{i\in [n]}$ be a full feasible collection of actions. Deciding whether there exists an action sequence $\pi$ that produces $A$ can be done in polynomial time. 
\end{theorem}

We prove Theorem~\ref{thm:as-for-F} using Algorithm~\ref{alg:feasible-collection-action-sequence}, which constructs the action sequence $\pi$ in a greedy way. It uses the set variable $P$ to keep the agents whose actions have not been checked yet; the set variable $M$ keeps the agent-action pairs with agents in $P$. The action sequence is stored in variable $\pi$, which is returned at the end. In case no action sequence can be computed, the algorithms returns a Fail report.

{
 	\begin{algorithm}[ht]
 		{ 

 			{\bf Input:} A feasibility constraint $F$, a feasible full collection of actions $\{i,a_i\}_{i\in[n]}$, rank information $\{r_i\}_{i\in [n]}$
 			
 			{\bf Output:} An action sequence that produces $\{i,a_i\}_{i\in[n]}$  or Fail report 
 			\caption{Deciding whether a full collection of actions can be produced by an action sequence}	
 			\label{alg:feasible-collection-action-sequence}
 			\begin{algorithmic}[1]
 				\STATE Initialize $\pi \leftarrow \emptyset$; $P\leftarrow [n]$; $M\leftarrow \emptyset$; $k\leftarrow 1$; 
 				\WHILE{$P\not=\emptyset$}
 				    \STATE $\text{Top}\leftarrow \{i\in P: a_i=\BR(i,M)\}$;
 				    \IF{$\text{Top}\not=\emptyset$} 
 				    \STATE Let $i^*\in \text{Top}$;
 				    \STATE $\pi(k)\leftarrow i^*$; $P\leftarrow P\setminus \{i^*\}$; $M\leftarrow M \cup \{(i^*,a_{i^*})\}$; $k\leftarrow k+1$;
 				    \ELSE
 				    \RETURN Fail;
 				    \ENDIF
 				\ENDWHILE
                \RETURN $\pi$;
   		\end{algorithmic}
 		}
 	\end{algorithm}
 }		
 
\begin{proof}
Clearly, the algorithm runs in polynomial time assuming we have polynomial-time access to $\BR(i,M)$. To prove its correctness, first notice that it indeed outputs an action sequence that produces the collection of actions $A$ if it never fails. Now, assume that $A$ can be produced by an action sequence $\widehat{\pi}$; we will show that the algorithm never fails. Indeed, at the beginning of any execution of the while loop, let $i$ be the agent in set $P$ that acts earliest in $\widehat{\pi}$. Denote by $\widehat{M}$ the collection of actions decided by the action subsequence $\widehat{\pi}^i$ (while $M$ is the collection of actions that have been checked prior to the current while loop execution. Since $M\cup\{(i,a_i)\} \subseteq A$, downward closedness of $F$ implies that $M\cup\{(i,a_i)\}$ is feasible. Furthermore, since $\widehat{M}\cup\{(i,x)\} \subseteq M\cup\{(i,x)\}$ for any $x\in X_i$, we conclude that any action that is available to agent $i$ in $\pi^i$ was also available in $\widehat{\pi}^i$. Endogeneity of valuations then implies that $a_i=\BR(i,M)$. 
\end{proof}

For \osm, Theorem~\ref{thm:as-for-F} implies the following corollary.
\begin{corollary}\label{cor:as-for-matching}
Given an instance of \osm\ and a perfect matching $\mu$ in the underlying bipartite graph $G$, deciding whether there exists an action sequence $\pi$ that produces $\mu$ can be done in polynomial time. 
\end{corollary}

We now use Algorithm~\ref{alg:feasible-collection-action-sequence} together with Corollary~\ref{cor:as-for-matching} to rediscover a characterization of Abdulkadiro\u{g}lu and S\"onmez~\cite{AS98} (see also~\cite{ACM+05}) about matchings that can be computed by serial dictatorships. To do so, we need to adapt Pareto-optimality to our case. 

\begin{definition}
Given a complete bipartite graph $G=K_{n,n}$ and rank information $\{r_i\}_{i\in [n]}$ representing strict preferences of agent $i\in [n]$ for $n$ items, a perfect matching $\mu$ is Pareto-optimal if there is no other perfect matching $\mu'$ such that $r_i(\mu'_i)\leq r_i(\mu_i)$ for all $i\in [n]$ and $r_i(\mu'_{i^*})< r_i(\mu_{i^*})$ for some agent $i^*\in [n]$.
\end{definition}

\begin{theorem}[Abdulkadiro\u{g}lu and S\"onmez~\cite{AS98}]\label{thm:pareto}
Given an instance of \osm\ consisting of a complete bipartite graph $G=K_{n,n}$ and rank information $\{r_i\}_{i\in [n]}$, a perfect matching can be produced by an action sequence if and only if it is Pareto-optimal.
\end{theorem}

\begin{proof}
We prove that a perfect matching $\mu$ can be accepted by Algorithm~\ref{alg:feasible-collection-action-sequence} (by selecting $i^*$ from $\text{Top}$ appropriately) if and only if $\mu$ is Pareto-optimal. Recall that, now, Algorithm~\ref{alg:feasible-collection-action-sequence} takes  $(i,\mu_i)_{i\in [n]}$ as input.

Assume that Algorithm~\ref{alg:feasible-collection-action-sequence} fails for some matching $\mu$. This means that $\text{Top}$ is empty after the execution of line 3 in some iteration of the while loop. Hence, for every agent $i$ in set $P$, $\mu_i\not=\nu_i=\BR(i,M)$. Consider the subgraph consisting of the nodes corresponding to the agents in $P$ and the items they get in $\mu$ and the edges $(i,\mu_i)$ and $(i,\nu_i)$ for $i\in P$. As this subgraph has $2|P|$ nodes and edges, it has a cycle $C$, which for every edge $(i,\mu_i)$ it contains, it also has edge $(i,\nu_i)$. Now, we define the matching $\mu'$ as follows. We set $\mu'_i=\nu_i$ if node $i$ belongs to the cycle $C$ and $\mu'_i=\mu_i$ otherwise. By the definition of $\nu_i$, every agent $i$ of $C$ satisfies $r_i(\mu'_i)<r_i(\mu_i)$ while $r_i(\mu'_i)=r_i(\mu_i)$ for any agent $i\in [n]\setminus C$. Hence, $\mu'$ Pareto-dominates $\mu$, proving that $\mu$ is not Pareto-optimal.

Now, consider a matching $\mu$ that is not Pareto-optimal; we will show that the algorithm will fail. Let $\mu'$ be a matching that Pareto-dominates $\mu$. Let $P_C$ be the set of agents who get a different item in $\mu$ and $\mu'$. Also, denote by $R_C$ the set of items the agents of $P_C$ get in $\mu$ and $\mu'$. Consider an execution of the while loop in a step with $P_C\subseteq P$ (so, no action of the agents in $P_C$ has been checked yet). Then, for every agent $i\in P_C$, it holds $r_i(\mu'_i)<r_i(\mu_i)$ and, furthermore, as an item of $R_C$, item $\mu'_i$ is not among the actions that have been checked so far. Hence, $\BR(i,M)\not=\mu_i$. This means that no agent of $P_C$ can ever be included in set $\text{Top}$ and the algorithm will fail before set $P$ becomes empty.
\end{proof}

To answer Question 3, it suffices to use Theorem~\ref{thm:pareto} and observe that there is always a maximum-weight perfect matching that is Pareto-optimal. In this way, we obtain the following corollary.

\begin{corollary}
The price of serial dictatorship of maximum weight perfect matching in bipartite graphs is $1$.
\end{corollary}

\section{Arborescences in Directed Graphs} \label{sec:arborescence}
We devote this section to \osa; we answer Question 1 using Algorithm~\ref{alg:arborescence-2-approx}.

{
	\begin{algorithm}[ht]
		{ 
			
			{\bf Input:} Query access to 
			an $n$-agent \osa\ instance 
			
			{\bf Output:} An action sequence $\pi$
			\caption{A $2$-approximation algorithm for \osa}	
	\label{alg:arborescence-2-approx}
			\begin{algorithmic}[1]
				\STATE Initialize  $\pi \leftarrow \emptyset$; $ B \leftarrow \emptyset$;
				\FOR{$i = 1$ to  $n$}
				\IF{$v_i(\pi) = v_i(\emptyset)$}
				\STATE Update $\pi \leftarrow \pi||i$;\\
				\ELSE
				\STATE $C \leftarrow \{i\} \cup \{ j \in \pi: v_i(\pi \setminus j) = v_i(\emptyset)\}$; \label{line7}\\
				\STATE $\ell \leftarrow \argmin_{t \in C} v_t(\emptyset)$;
				\STATE Update $\pi \leftarrow (\pi || i) \setminus \ell$; $B \leftarrow B || j$; 
				\ENDIF
				\ENDFOR \\
				\STATE Update $\pi \leftarrow \pi ||B$;
				\RETURN  $\pi$;
			\end{algorithmic}
		}
	\end{algorithm}
}	

Algorithm~\ref{alg:arborescence-2-approx} builds the action sequence $\pi$ gradually. It aims to put at the beginning of $\pi$ agents who get their maximum value (i.e., equal to $v_i(\emptyset)$) by their action. The actions of these agents form a directed forest in the underlying directed graph. To do so, the algorithm considers the agents one by one (in the for loop) and, when an agent does not get her maximum value (lines 6-8) because the corresponding edge would form a cycle in the current forest, it detects this cycle (through the queries to $v_i(\pi\setminus j)$ and $v_i(\emptyset)$ in line 6) and removes the agent of lowest value (i.e., the one who has drawn the edge of lowest weight) in it from $\pi$ (in lines 7 and 8). Agents who are removed during this process are put in a reserve action subsequence $B$ (line 8). When all agents have been considered, the final action sequence consists of $\pi$ followed by $B$ (line 11).

\begin{theorem} \label{theorem:2-arborescence}
	 On input $n$-agent \osa\ instances, Algorithm~\ref{alg:matching-2-approx} uses $O(n^2)$ queries and computes a $2$-approximate solution.
\end{theorem}

\begin{proof}
The bound on the running time follows since there are $n$ executions of the for loop and each execution requires at most $O(n)$ queries. In particular, at most $n$ queries are enough to compute $C$ in line 6 and, similarly, at most $n$ queries are enough to compute $\ell$ in line 7.

To prove the bound on the approximation ratio, we will account for the contribution to the social welfare by the agents in action subsequence $\pi$ that is computed by the for loop (before the update at line 11). By the definition of the algorithm and due to monotonicity, each of these agents $i$ contributes $v_i(\emptyset)$. We use the term {\em top edge} to refer to the edge these agents draw.

Let $N$ be the set of agents that are put to the action subsequence $B$ and let $P$ the all other agents. For an agent $\ell\in N$, let $i$ be the agent considered in the for loop when $\ell$ was removed from $\pi$. This removal happened because the top edges of the agents already in $\pi$ form a cycle $C$ together with the top edge of agent $i$, which has minimum weight among all edges in $C$.

Let $f(\ell)$ be an arbitrary agent different than $\ell$ in $C$. By the definition of Algorithm~\ref{alg:arborescence-2-approx}, it holds that  $v_\ell(\emptyset)\leq v_{f(\ell)}{\emptyset}$. Furthermore, observe that $f(\ell)$ cannot appear in any cycle formed in later executions. This is due to the fact that the only nodes that can be reached from $f(\ell)$ through edges already drawn or edges that will be drawn in later executions are those in the path from $f(\ell)$ to $\ell$ (whose drawn edge was just removed). Hence, $f(\ell)\in P$ and, furthermore, $f(\ell)$ is a one-to-one function from $N$ to $P$. Thus,
\begin{align*}
    \sum_{i\in P}{v_i(\emptyset)} &\geq \sum_{\ell\in N}{v_{f(\ell)}(\emptyset)}\geq \sum_{i\in N}{v_i(\emptyset)}
\end{align*}
and
\begin{align*}
    \SW(\pi) &\geq \sum_{i\in P}{v_i(\pi^i)} =\sum_{i\in P}{v_i(\emptyset)} \geq \frac{1}{2}\sum_{i\in [n]}{v_i(\emptyset)}\geq \frac{1}{2}\cdot \SW(\widehat{\pi}),
\end{align*}
where $\widehat{\pi}$ denotes the optimal action sequence.
\end{proof}

Like \osm, the answer to Question 2 for \osa\ follows by Theorem~\ref{thm:as-for-F}; recall that the feasibility constraint of a directed forest is downward closed and valuations are endogenous.

\begin{corollary}\label{cor:as-for-arborescence}
Given an instance of \osa\ and an arborescence $\tau$ in the underlying directed graph $G$, deciding whether there exists an action sequence $\pi$ that produces $\tau$ can be done in polynomial time. 
\end{corollary}

We now characterize the arborescences that can be produced by action sequences using, again, the notion of Pareto-optimality.

\begin{definition}
Given a complete directed graph $G=([n],E)$ and rank information $\{r_i\}_{i\in [n]}$ representing strict preferences of agent $i\in [n]$ for the edges incident to node $i$, an arborescence $\tau$ is Pareto-optimal if there is no other arborescencce $\tau'$ such that $r_i(\tau'_i)\leq r_i(\tau_i)$ for all $i\in [n]$ and $r_i(\tau'_{i^*})< r_i(\tau_{i^*})$ for some agent $i^*\in [n]$.
\end{definition}

\begin{theorem}\label{thm:pareto-arb}
Given an instance of \osa\ consisting of a directed graph $G=([n],E)$ and rank information $\{r_i\}_{i\in [n]}$, an arborescence $\tau$ can be produced by an action sequence if and only if $\tau$ is Pareto-optimal.
\end{theorem}

\begin{proof}
Using Corollary~\ref{cor:as-for-arborescence}, it suffices to show that an arborescence $\tau$ can be accepted by Algorithm~\ref{alg:feasible-collection-action-sequence} if and only if it is Pareto-optimal.

Assume that Algorithm~\ref{alg:feasible-collection-action-sequence} fails for some arborescence $\tau$. This means that $\text{Top}$ is empty after the execution of line 3 in some iteration of the while loop. Hence, for every agent $i\in P$, it is $\tau_i\not=\BR(i,M)$. Consider a node $k\in P$ that is furthest from the root node in $\tau$. Let $\nu_k=\BR(k,M)$ and assume that $\nu_k$ belongs to the subtree of $k$ in $\tau$. Since $\nu_k=\BR(k,M)$, $M\cup \{(k,\nu_k)\}$ is feasible. Hence, some edge $(k',\tau_{k'})$ in the path from $\nu_k$ to $k$ has not been checked yet by the algorithm. Therefore, $k'$ should belong to $P$. This contradicts the assumption that $k$ is a node that is furthest from the root in $\tau$. Hence, $\nu_k$ does not belong to the subtree of $k$ in $\tau$. By replacing the edge $(k,\mu_k)$ by $(k,\nu_k)$ in $\tau$, we obtain an arborescence that Pareto-dominates $\tau$ (recall that $r_k(\nu_k)<r_k(\tau_k)$).

Now, consider an arborescence $\tau$ that is not Pareto-optimal; we will show that Algorithm~\ref{alg:feasible-collection-action-sequence} fails for it. Let $\tau'$ be an arborescence that Pareto-dominates $\tau$ and let $k$ be an agent who uses different actions in $\tau$ and $\tau'$ and is furthest from the root in $\tau$. Consider the execution of the while loop in a step with $k\in P$. Repeating the argument in the previous paragraph, we obtain that $\tau'_k$ cannot be in the subtree of $k$ in $\tau$. Hence, the collection of actions $M\cup \{(k,\tau'_k)\}$ is feasible and, furthermore, $r_k(\tau'_k)<r_k(\tau_k)$. Thus, $\BR(k,M)\not=\tau_k$ and agent $k$ can never be included in $\text{Top}$ before Algorithm~\ref{alg:feasible-collection-action-sequence} fails.
\end{proof}

As there is always a maximum-weight arborescence that is Pareto-optimal, we obtain the following corollary.

\begin{corollary} \label{cor:PoSD-arb}
The price of serial dictatorship of maximum weight arborescence in complete graphs is $1$.
\end{corollary}

\section{Satisfiability} \label{sec:sat}
\osm\ and \osa\ enjoy many similarities; this is evident from our answers to the three key questions. In contrast, the case of \oss\ is fundamentally different. 

\begin{theorem} \label{theorem:2-sat}
Any action sequence yields a $2$-approximate solution to \oss.
\end{theorem}
\begin{proof}
Consider an action sequence $\pi$. Then, the quantity $P(i,C_{\pi^i})+N(i,C_{\pi^i})$ is the total weight of the clauses in which variable $x_i$ appears and which are unsatisfied when it is the turn of agent $i$ to act. Hence, $\sum_{i\in [n]}{\left(P(i,C_{\pi^i})+N(i,C_{\pi^i})\right)}$ is an upper bound on the total weight in all clauses. In addition, the value agent $i$ gets is clearly $v_i(\pi^i)\geq \frac{1}{2}\left(P(i,C_{\pi^i})+N(i,C_{\pi^i})\right)$. Thus,
\begin{align*}
\SW(\pi) &= \sum_{i \in [n]}{v_i(\pi^i)} \geq \frac{1}{2} \sum_{i \in [n]}{\left(P(i,C_{\pi^i})+N(i,C_{\pi^i})\right)} \geq \frac{1}{2} \sum_{j \in [m]}{w(C_j)}\geq \frac{1}{2}\SW(\widehat{\pi}), 
\end{align*}
where $\widehat{\pi}$ denotes the optimal action sequence for the \oss\ instance.
\end{proof}

Given Theorem~\ref{theorem:2-sat}, the obvious next step would be to design more sophisticated algorithms for achieving a better-than-$2$ approximation ratio with polynomially many queries. However, we believe that there are important obstacles in doing so for \oss, summarized in the next question:

\smallskip

\begin{quote}\begin{mdframed}
{\bf Open Question 2:} Is there an exponential lower bound on the query complexity of (approximating) \oss?
\end{mdframed}\end{quote}

\smallskip

Next, we denote by \satas\  the problem of deciding whether a given Boolean assignment for the underlying satisfiability instance of \oss\ can be produced by an action sequence. In contrast to the easiness of deciding whether a given perfect matching or arborescence can be produced by an action sequence, \satas\ turns out to be intractable.

\begin{theorem}
	\satas\ is NP-hard.
\end{theorem}

\begin{proof}
We prove the theorem by presenting a polynomial-time reduction from {\sc Exact-3Cover} (X3C). An instance of X3C consists of a universe $\mathcal{U}$ of $3q$ elements and a collection $\mathcal{T}$ of $t$ sets that each contains exactly three elements from $\mathcal{U}$. The problem of deciding whether there exists an {\em exact 3-cover}, i.e., a subcollection of $q$ sets from $\mathcal{T}$ that includes all elements of $\mathcal{U}$, is a well-known NP-hard problem~\cite{GJ79}.

For an element $x \in \mathcal{U}$, we write $f(x) = |T \in \mathcal{T}: x \in T|$ to denote the frequency of $x$ in collection $\mathcal{T}$. We construct a satisfiability instance which contains one {\em element variable} for each element in $\mathcal{U}$, one {\em set variable} for each set in $\mathcal{T}$, and one auxiliary variable $Q$. Additionally, there are five types of clauses, each of which contains at most two literals:
	\begin{itemize}
		\item Type A: Clause $(T_i \vee T_j)$ of weight $1$ for each pair of distinct sets $T_i, T_j \in \mathcal{T}$. There are a total of $t(t-1)/2$ such clauses.
		\item Type B: Clauses $(T \vee \overline{x}), (T \vee \overline{y}),$ and $(T \vee \overline{z})$, each of weight $1$, for each set $T = \{x,y,z\} \in \mathcal{T}$. There are a total of $3t$ such clauses.
		\item Type C: Clause $(x \vee \overline{Q})$  of weight $f(x)-1/3$ for each element $x \in \mathcal{U}$. There are a total of $3q$ such clauses.
		\item Type D: Clauses $(Q \vee \overline{T_1}), (Q \vee \overline{T_2}), \dots,  (Q \vee \overline{T_t})$, each of weight $t-q+7/3$. There are a total of $t$ such clauses.
		\item Type E: Clause (with singleton literal) $\overline{Q}$ of weight $t^2-qt+7t/3-1/3$.
	\end{itemize}
	Hence, the satisfiability instance consists of $n=3q+t+1$ variables and $m= t^2/2+7t/2+3q+1$ clauses. We now show that there is an action sequence which assigns the value of True to each variable of the satisfiability instance if and only if the X3C instance has an exact 3-cover. The next three lemmas state important properties an action sequence must have so that all variables are set to True. We use the terms {\em set agent} and {\em element agent} to refer to the agents who control the corresponding variables.
	
\begin{lemma} \label{lem:Q-element}
Agent $Q$ should act after all element agents.
\end{lemma}

\begin{proof}
Notice that variable $Q$ appears as positive literal in clauses of total weight $t^2 -qt +7t/3$. Denoting by $w_C$ the total weight of all type C clauses, we have that $Q$ appears as negative literal in clauses of total weight $t^2 -qt +7t/3-1/3+w_C$. Since $w_C \geq 2/3$, agent $Q$ can act and set its variable to True only after all Type C clauses are satisfied, which implies that all element agents have acted before (and set their variable to True).
\end{proof}

\begin{lemma}\label{lem:x-T}
Element agent $x$ should act after at least one set agent $T$ such that $x \in T$ has acted.
\end{lemma}

\begin{proof}
Element variable $x$ appears as positive and negative literal in clauses of total weight $f(x) -1/3$ and $f(x)$, respectively. As literal $\overline{x}$ appears only in type B clauses, the element agent $x$ can set variable $x$ to True only if at least one set agent $T$ such that $x \in T$ acted before (and set her variable to True).
\end{proof}

\begin{lemma}\label{lem:q-Q}
At most $q$ many set agents can act before agent $Q$.
\end{lemma}

\begin{proof}
Consider the time when $q$ set agents have acted; denote them by set $R$. Every set agent $T$ not belonging to $R$ appears as positive and negative literal in clauses of total weight $t-q+2$ and $t-q+7/3$, respectively, and would set their variable to False if they acted before agent $Q$ acts (and decreases the total weight of unsatisfied clauses that include $T$ as a negative literal).
\end{proof}

Lemmas~\ref{lem:Q-element},~\ref{lem:x-T}, and~\ref{lem:q-Q} establish that the action sequence that assigns each variable to True must have the following form: it starts with a set $R$ of at most $q$ set agents
corresponding to an exact 3-cover of the X3C instance, then come all the $3q$ element agents, followed by agent $Q$, and finally the remaining set agents. We can verify that such a structure works for any $R$ that form an exact 3-cover, completing the proof.
\end{proof}

Finally, we show that the price of serial dictatorship for maximum satisfiability is strictly higher than $1$.

\begin{theorem} \label{theorem: sat SD}
There exists an \oss\ instance in which  no action sequence produces better than $3/2$-approximate solution to the underlying satisfiability instance.
\end{theorem}

\begin{proof}
	Let $\varepsilon>0$ be a negligibly small value. Consider the \oss\ instance defined by three Boolean variables $x_1, x_2, x_3$ and the following six clauses: $(x_1 \vee \overline{x_2} \vee \overline{x_3})$, $(\overline{x_1} \vee x_2 \vee \overline{x_3})$, $(\overline{x_1} \vee \overline{x_2} \vee x_3)$, each of weight $1$, and $(x_1)$, $(x_2)$, $(x_3)$, each of weight $1- \varepsilon$. It is easy to verify that any action sequence sets the first two variables to $0$ and the third one to $1$, producing an assignment that satisfies clauses of total weight at most $4-\varepsilon$ while the assignment $x_1=x_2=x_3=1$ has maximum total weight of $6-3\varepsilon$.
\end{proof}

Closing the gap between $3/2$ and $2$ is an important open problem.
\smallskip

\begin{quote}\begin{mdframed}
{\bf Open Question 3:} What is the tight bound on the price of serial dictatorship for maximum satisfiability?
\end{mdframed}\end{quote}

\smallskip

Better upper bound on the price of serial dictatorship for maximum satisfiability should exploit simple greedy algorithms. An algorithm we have studied considers the agents in monotone non-increasing order in terms of the total weight in singleton clauses containing their variable. Proving that this algorithm has an approximation ratio, say, $3/2$, would be enough to close the PoSD gap. Notice that this would match a known $3/2$ bound from~\cite{CFZ97} for a slightly more complicated greedy algorithm proposed by Johnson~\cite{Johnson74}. 

%% file: truthful-implementation.tex
\section{Truthful Implementation} \label{sec:truth}
In this section, we explore truthful implementations of algorithms for \pro. A serial dictatorship, that uses a fixed action sequence, is the simplest truthful mechanism for \pro. As algorithms exploit queries to the agents in the computation of action sequences, the agents may have incentives to misreport their valuations and benefit from the output of the algorithm. Payments can be used in the design of sophisticated algorithms to ensure that responses to the queries are sincere.

In this section, we use the terms mechanism and algorithm interchangeably. In general, a mechanism for \pro\ takes as input (part of) the valuations of the agents and decides an action sequence {\em outcome} and {\em payments} that are imposed to the agents. Formally, we write $\bv=\{v_i\}_{i\in [n]}$ to denote the valuations of all agents, while $\bv_{-i}$ is the restriction of $\bv$ to the values of all agents besides $i$. For a mechanism, we typically use $M(\bv)$ to denote the action sequence computed on input (part of) the valuation vector $\bv$.\footnote{
Of course, we assume that any agent $i$ can compute her valuations $v_i(S)$ for any action subsequence $S$. For the specific valuations structures, this implies that agent $i$ can simulate the actions of agents appearing in $S$, and thus can decide the best available action for her. That is, we assume a complete information environment for the agents, as opposed to the mechanism which relies on queries to elicit information about valuations.} As usual, the notation $M^i(\bv)$ is used to denote the prefix of agent $i$ in the action sequence $M(\bv)$. Also, we denote by $p_i(\bv)$ the payment imposed to agent $i\in [n]$. Then, the quantity  $v_i(M^i(\bv))-p_i(\bv)$ is simply the utility of agent $i$ when agents report the valuation vector $\bv$.
Truthfulness requires that the mechanism returns the action sequence $M(\bv)$ and payments $p_i(\bv)$ for $i\in [n]$ so that no agent $i$ has an incentive to misreport any false valuation $v'_i$ instead of her true valuation $v_i$, for any valuations $\bv_{-i}$ by the other agents, i.e., 
\begin{align*}
    v_i(M^i(\bv))-p_i(\bv) &\geq v_i(M^i(\bv_{-i},v'_i))-p_i(\bv_{-i},v'_i).
\end{align*}
In other words, truth-telling is a utility-maximizing strategy for agent $i$, no matter what strategies the other agents use. To adapt this definition to randomized mechanisms (in which either the action sequence or the payment part are random), it suffices to use expectations of all terms.

In the following, we would like to explore whether the algorithms for \pro\ we have already designed are truthful for appropriately defined payments. In other words, we explore whether they have truthful implementations. A useful tool here might be a well-known {\em cycle monotonicity} characterization result of Rochet~\cite{R87}. However, the following much simpler argument is enough to prove negative results.

\begin{lemma}\label{lem:cycle-mon}
An algorithm $M$ has a truthful implementation only if for any valuation profile $\bv=(\bv_{-i},v_i)$ and an alternative valuation $v'_i$ that agent $i$ may report, the following is true:
\begin{align} \label{eq:truthful}
v_i(M^i(\bv))+v'_i(M^i({\bv_{-i}},v'_i)) \geq v_i(M^i({\bv_{-i}},v'_i))+v'_i(M^i({\bv})).
\end{align}
\end{lemma}

\begin{proof}
Fix an agent $i \in [n]$. By truthfulness, we know that an agent $i$ has no incentive to misreport her value from $v_i$ to $v'_i$ and vice versa. That is, assuming payment $p_i$, 
\begin{align*}
    &v_i(M^i(\bv))- p_i(\bv) \geq v_i(M^i(\bv_{-i},v'_i))-p_i(\bv_{-i}, v'_i) \\
    &v'_i(M^i(\bv_{-i},v'_i))- p_i(\bv_{-i}, v'_i) \geq v'_i(M^i(\bv))-p_i(\bv).
\end{align*}
Summing the inequalities above, we get the desired inequality.
\end{proof}

We can now use Lemma~\ref{lem:cycle-mon} to prove some non-implementability results.

\begin{theorem}
Algorithms~\ref{alg:matching-2-approx},~\ref{alg:arborescence-2-approx}, and {\sc Det} do not have truthful implementations.
\end{theorem}

\begin{proof}
We present instances of \osm, \osa, and \pro\ for which Algorithms~\ref{alg:matching-2-approx},~\ref{alg:arborescence-2-approx}, and {\sc Det}, respectively, do not satisfy the inequality~(\ref{eq:truthful}) in Lemma~\ref{lem:cycle-mon} and, thus, cannot be implemented truthfully.

Consider the \osm\ instance with two agents and two items and the following weights: $w(1,1)=1+\varepsilon$, $w(1,2) = 1$ defining the valuation $v_1$ for agent $1$, and $w(2,1)=1$, $w(2,2) = 1-\varepsilon$ defining the valuation $v_2$ for agent $1$. Algorithm~\ref{alg:matching-2-approx} (implicitly) matches agent $1$ to item $1$ and agent $2$ to item $2$. Now, consider another valuation profile $(v'_1,v_2)$ corresponding to the weights $w(1,1)=1-\varepsilon$, $w(1,2) = 0$, $w(2,1)=1$, $w(2,2) = 1-\varepsilon$, and observe that Algorithm~\ref{alg:matching-2-approx} matches agent $2$ to item $1$ and agent $1$ to item $2$. Denoting the outcome of Algorithm~\ref{alg:matching-2-approx} by $M$, we have $v_1(M^1(\bv))=1+\varepsilon$ and $v'_1(M^1(\bv_{-1},v'_1))=0$, while $v_1(M^1({\bf v_{-1}},v'_1)=1$ and $v'_1(M^1({\bf v}))=1-\varepsilon$. These values violate inequality~(\ref{eq:truthful}).

Next, consider the $4$-agent \osa\ instance with the following edge-weights: $w(1,2)=1-\varepsilon$, $w(1,4) =\varepsilon$, $w(2,1)=1$, $w(3,4)=1-\varepsilon$, and $w(4,3)=1$, with all the remaining edges having weight $0$. We denote the corresponding valuation profiles of the three agents by $v_1$, $v_2, v_3$ and $v_4$. Also, consider the following instance where the edge-weights are $w(1,2)=1+\varepsilon$, $w(1,4) =1$, $w(2,1)=1$, $w(3,4)=1-\varepsilon$, and $w(4,3)=1$, with all the remaining edges having weight $0$. We denote the corresponding valuation profiles of the three agents by $v'_1$, $v_2, v_3$ and $v_3$. We have $v_1(M^1(\bv))=0$ and $v'_1(M^1(\bv_{-1},v'_1))=1+\varepsilon$, while $v_1(M^1({\bf v_{-1}},v'_1))=1-\varepsilon$ and $v'_1(M^1({\bf v}))=1$, again violating inequality~(\ref{eq:truthful}).

Finally, consider the instance of \pro\ with $n$ agents having the following valuations: there is a set $C$ of $c$ agents that have $v_i(S)=10$ if $|S|<c$ and $v_i(S)=8$ otherwise. The rest of the agents have valuation $v_i(S)=9$ for any $S$. So, the action sequence $M({\bf v})$ has the agents in $C$ in the first $c$ positions, with the agents of $[n]\setminus C$ next. For some agent $j\in  C$, the valuations $v'_j$ are defined as $v'_j(S)=8$ if $|S|<c$ and $v'_j(S)=0$, otherwise. The action sequence $M({\bf v_{-j}},v'_j)$ has the agents of $C\setminus\{j\}$ and an additional agent from $[n]\setminus C$ in the first $c$ positions, and the rest of the agents (including $j$) in the last $n-c$ positions.
We have $v_j(M^j({\bf v}))=10$ and $v_j(M^j({\bf v_{-j}},v'_j)=8$, while $v'_j(M^j({\bf v}))=8$ and $v'_j(M^j({\bf v_{-j}},v'_j))=0$, again violating inequality~(\ref{eq:truthful}).
\end{proof}

On the positive side, {\sc Rand} can be truthfully implemented using VCG payments. Recall that {\sc Rand} selects a random set $C$ of $c\leq n$ agents and puts them in the beginning of the action sequence in such a way that their contribution to the social welfare is maximized. By adapting the VCG payment scheme, we can define payments as $p_i(\bv)=0$ if $i$ is not selected in set $C$ and
\begin{align*}
    p_i(\bv) &= \sum_{k\in C\setminus\{i\}}{v_k(M^k(\bv_{-i},{\bf 0}))}-\sum_{k\in C\setminus\{i\}}{v_k(M^k(\bv))}
\end{align*}
otherwise, where $\bf 0$ denotes valuation of zero for agent $i$ for any action subsequence. In this way, {\sc Rand} enjoys properties like non-negative payments and {\em individual rationality} (i.e., agents enjoy non-negative utilities).

We will now show how to use the VCG paradigm in a variation of {\sc Det}, to get the truthful mechanism {\sc Det+} for \pro. Like {\sc Det}, {\sc Det+} uses an integer parameter $c$. It considers all action sequences which have their last $n-c$ agents ordered according to their indices and picks the one of maximum social welfare (in this way, the outcome of {\sc Det+} is at least as good as that of {\sc Det}). VCG payments, defined as
\begin{align*}
    p_i(\bv) &= \sum_{k\in [n]\setminus \{i\}}{v_k(M^k(\bv_{-i},{\bf 0}))}-\sum_{k\in [n]\setminus \{i\}}{v_k(M^k(\bv))},
\end{align*}
lead to a truthful implementation for {\sc Det+} (with non-negative payments and individual rationality). The next statement summarizes the above discussion on {\sc Rand} and {\sc Det+}.

\begin{theorem}
Algorithms {\sc Rand} and {\sc Det+} have truthful implementations. Also, {\sc Det+} achieves the same approximation guarantee with {\sc Det}.
\end{theorem}

We remark that, in the definition of the payments, it is important to consider agent $i$ with zero valuation. Essentially, the VCG payment of agent $i$ is equal to the harm that her non-zero valuation (as opposed to her existence, which is the usual interpretation) causes to the other agents. Furthermore, we remark that the payments for agents in the last $n-c$ positions of $M(\bv)$ is not necessarily $0$. 

In addition to the better approximation ratio that {\sc Rand} achieves compared to {\sc Det+}, another advantage is that payments require only polylogarithmically many queries to compute (when the parameter $c$ is set to $\Theta\left(\frac{\ln{n}}{\ln\ln{n}}\right)$). On the other hand, the VCG payments for {\sc Det+} require $\Theta(n^2)$ queries.\footnote{In comparison to random serial dictatorship (RSD), which selects an action sequence from $\Sn$ uniformly at random, the truthful algorithms {\sc Rand} and {\sc Det+} achieve sublinear approximation ratio at the expense of using payments. We can show that this is necessary: any (possibly randomized) mechanism for \pro\ that is truthful without payments has an approximation ratio of at least $n$.}
\smallskip

\begin{quote}\begin{mdframed}
    {\bf Open Question 4:} Is there a deterministic algorithm for \pro\ that uses polynomially many queries, achieves sublinear approximation ratio, and has a truthful implementation with payments that can be computed with polylogarithmically many queries?
\end{mdframed}\end{quote}
\smallskip

Regarding \pro\ instances with specific valuation structures, recall that
\oss\ has a truthful mechanism that does not use payments at all. Indeed, any serial dictatorship is truthful and furthermore achieves an approximation ratio of $2$. We will present a similar result for \osa. In particular, we develop the randomized algorithm {\sc Bit} that is truthful without payments and computes $2$-approximate arborescences. {\sc Bit} tosses a fair coin and returns the action sequence $(1,2, ..., n)$ on {\sc Heads}, and the reverse action sequence $(n,n-1, ..., 1)$ on {\sc Tails}. 

\begin{theorem} \label{thm:bit}
Algorithm {\sc Bit} is truthful and returns $2$-approximate solutions of \osa.\footnote{We remark that the argument in the proof of Theorem~\ref{thm:bit} can be used to show that RSD returns a $2$-approximate solution too. Obviously, {\sc Bit} is superior since it uses only a single random coin.}
\end{theorem}

\begin{proof}
Algorithm~{\sc Bit} is clearly truthful since it does not use any queries. We will account for the contribution of an agent to the social welfare only if/when she gets a value of $v_i(\emptyset)$. Consider an agent $i\in [n]$ and let $(i,k)$ be her rank-$1$ edge. Notice that agent $i$ will draw edge $(i,k)$ if she acts before agent $k$, which happens with probability $1/2$. Hence, $\SW(\pi)=\sum_{i\in [n]}{v_i(\pi^i)}\geq \frac{1}{2}\sum_{i\in [n]}{v_i(\emptyset)}=\frac{1}{2}\SW(\widehat{\pi})$, where $\widehat{\pi}$ is the optimal action sequence.
\end{proof}

\begin{quote}\begin{mdframed}
    {\bf Open Question 5:} Is there an \osm\ algorithm that uses polynomially many queries, achieves constant approximation ratio, and has a truthful implementation?
\end{mdframed}\end{quote}
\smallskip

Recall that the existence of an algorithm that solves \osm\ exactly is open. In such a case, VCG payments would provide a truthful implementation.

Finally, regarding \oss, Theorem~\ref{theorem: sat SD} implies that any serial dictatorship is a $2$-aproximate truthful mechanism without payments.

%% file: main.bbl
\newcommand{\etalchar}[1]{$^{#1}$}
\begin{thebibliography}{CFRF{\etalchar{+}}16}

\bibitem[ACMM05]{ACM+05}
David~J. Abraham, Katar{\'i}na Cechl{\'a}rov{\'a}, David~F. Manlove, and Kurt
  Mehlhorn.
\newblock Pareto optimality in house allocation problems.
\newblock In {\em Proceedings of the 16th International Symposium on Algorithms
  and Computation (ISAAC)}, pages 3--15, 2005.

\bibitem[ADK{\etalchar{+}}08]{ADK+08}
Elliot Anshelevich, Anirban Dasgupta, Jon~M. Kleinberg, {\'{E}}va Tardos, Tom
  Wexler, and Tim Roughgarden.
\newblock The price of stability for network design with fair cost allocation.
\newblock {\em {SIAM} Journal on Computing}, 38(4):1602--1623, 2008.

\bibitem[AS98]{AS98}
Atila Abdulkadiro\u{g}lu and Tayfun S\"onmez.
\newblock Random serial dictatorship and the core from random endowments in
  house allocation problems.
\newblock {\em Econometrica}, 66(3):689--701, 1998.

\bibitem[BBLM10]{BBL+10}
Allan Borodin, Joan Boyar, Kim~S. Larsen, and Nazanin Mirmohammadi.
\newblock Priority algorithms for graph optimization problems.
\newblock {\em Theorerical Computer Science}, 411(1):239–258, 2010.

\bibitem[BFT11]{BFT11}
Dimitris Bertsimas, Vivek~F. Farias, and Nikolaos Trichakis.
\newblock The price of fairness.
\newblock {\em Operations Research}, 59(1):17--31, 2011.

\bibitem[BM01]{BM01}
Anna Bogomolnaia and Hervé Moulin.
\newblock A new solution to the random assignment problem.
\newblock {\em Journal of Economic Theory}, 100(2):295--328, 2001.

\bibitem[BNR03]{BNR03}
Allan Borodin, Morten~N. Nielsen, and Charles Rackoff.
\newblock (incremental) priority algorithms.
\newblock {\em Algorithmica}, 37(4):295--326, 2003.

\bibitem[CFRF{\etalchar{+}}16]{CFF+16}
Ioannis Caragiannis, Aris Filos-Ratsikas, S{\o}ren Kristoffer~Stiil
  Frederiksen, Kristoffer~Arnsfelt Hansen, and Zihan Tan.
\newblock Truthful facility assignment with resource augmentation: An exact
  analysis of serial dictatorship.
\newblock In {\em Proceedings of the 12th Conference on Web and Internet
  Economics (WINE)}, pages 236--250, 2016.

\bibitem[CFZ97]{CFZ97}
Jianer Chen, D.K. Friesen, and Hao Zheng.
\newblock Tight bound on johnson's algorithm for max-sat.
\newblock In {\em Proceedings of the 12th Annual IEEE Conference on
  Computational Complexity (CCC)}, pages 274--281, 1997.

\bibitem[CKKK12]{CKKK12}
Ioannis Caragiannis, Christos Kaklamanis, Panagiotis Kanellopoulos, and Maria
  Kyropoulou.
\newblock The efficiency of fair division.
\newblock {\em Theory of Computing Systems}, 50(4):589--610, 2012.

\bibitem[Cla71]{C71}
Edward~H. Clarke.
\newblock Multipart pricing of public goods.
\newblock {\em Public Choice}, 11:17--33, 1971.

\bibitem[CLRS09]{CLRS01}
Thomas~H. Cormen, Charles~E. Leiserson, Ronald~L. Rivest, and Clifford Stein.
\newblock {\em Introduction to Algorithms}.
\newblock The MIT Press, 3rd edition, 2009.

\bibitem[DI04]{DI04}
Sashka Davis and Russell Impagliazzo.
\newblock Models of greedy algorithms for graph problems.
\newblock In {\em Proceedings of the 15th Annual ACM-SIAM Symposium on Discrete
  Algorithms (SODA)}, page 381–390, 2004.

\bibitem[DMS17]{DMS17}
Argyrios Deligkas, George~B. Mertzios, and Paul~G. Spirakis.
\newblock The computational complexity of weighted greedy matching.
\newblock In {\em Proceedings of the 31st AAAI Conference on Artificial
  Intelligence (AAAI)}, page 466–472, 2017.

\bibitem[DV12]{DV12}
Shahar Dobzinski and Jan Vondr{\'{a}}k.
\newblock From query complexity to computational complexity.
\newblock In {\em Proceedings of the 44th Symposium on Theory of Computing
  Conference (STOC)}, pages 1107--1116, 2012.

\bibitem[FRFZ14]{FFZ:14}
Aris Filos-Ratsikas, S{\o}ren Kristoffer~Stiil Frederiksen, and Jie Zhang.
\newblock Social welfare in one-sided matchings: Random priority and beyond.
\newblock In {\em Proceedings of the 7th International Symposium on Algorithmic
  Game Theory (SAGT)}, pages 1--12, 2014.

\bibitem[GJ79]{GJ79}
M.~R. Garey and D.~S. Johnson.
\newblock {\em Computers and Intractability: A Guide to the Theory of
  NP-Completeness (Series of Books in the Mathematical Sciences)}.
\newblock W. H. Freeman, 1979.

\bibitem[GLW21]{gourves2021fairness}
Laurent Gourv{\`e}s, Julien Lesca, and Ana{\"e}lle Wilczynski.
\newblock On fairness via picking sequences in allocation of indivisible goods.
\newblock In {\em Proceedings of the 7th International Conference on
  Algorithmic Decision Theory (ADT)}, pages 258--272, 2021.

\bibitem[Joh74]{Johnson74}
David~S. Johnson.
\newblock Approximation algorithms for combinatorial problems.
\newblock {\em Journal of Computer and System Sciences}, 9(3):256--278, 1974.

\bibitem[KMRZ19]{krysta2014size}
Piotr Krysta, David~F. Manlove, Baharak Rastegari, and Jinshan Zhang.
\newblock Size versus truthfulness in the house allocation problem.
\newblock {\em Algorithmica}, 81(9):3422--3463, 2019.

\bibitem[LLN06]{LLN06}
Benny Lehmann, Daniel Lehmann, and Noam Nisan.
\newblock Combinatorial auctions with decreasing marginal utilities.
\newblock {\em Games and Economic Behavior}, 55(2):270--296, 2006.

\bibitem[Man13]{M13}
David~F. Manlove.
\newblock {\em Algorithmics of Matching Under Preferences}.
\newblock World Scientific, 2013.

\bibitem[PSWZ17]{PSW+17}
Matthias Poloczek, Georg Schnitger, David Williamson, and Anke Zuylen.
\newblock Greedy algorithms for the maximum satisfiability problem: Simple
  algorithms and inapproximability bounds.
\newblock {\em SIAM Journal on Computing}, 46:1029--1061, 01 2017.

\bibitem[Roc87]{R87}
Jean-Charles Rochet.
\newblock A necessary and sufficient condition for rationalizability in a
  quasi-linear context.
\newblock {\em Journal of Mathematical Economics}, 16(2):191--200, 1987.

\bibitem[RS90]{RS90}
Alvin~E. Roth and Marilda A.~Oliveira Sotomayor.
\newblock {\em Two-Sided Matching: A Study in Game-Theoretic Modeling and
  Analysis}.
\newblock Cambridge University Press, 1990.

\bibitem[Sve99]{SVE:99}
Lars-Gunnar Svensson.
\newblock {Strategy-proof allocation of indivisble goods}.
\newblock {\em Social Choice and Welfare}, 16(4):557--567, 1999.

\bibitem[Yao77]{Y77}
Andrew Chi-Chin Yao.
\newblock Probabilistic computations: Toward a unified measure of complexity.
\newblock In {\em Proceedings and the 18th Annual Symposium on Foundations of
  Computer Science (FOCS)}, pages 222--227, 1977.

\end{thebibliography}
